\documentclass[11pt]{eptcs}
\usepackage{breakurl}

\usepackage{amsmath}
\usepackage{amssymb}
\usepackage{amsfonts}
\usepackage{amsthm}
\usepackage{tikz}

\newcommand{\mC}[1]{\mathit{#1}}
\newcommand{\mF}[1]{\mathit{#1}}
\newcommand{\Set}{\mC{Set}}
\newcommand{\Id}{\mF{Id}}
\newcommand{\C}{\mC{C}}
\newcommand{\F}{\mF{F}}
\newcommand{\B}{\mF{B}}
\newcommand{\I}{\mC{I}}

\usetikzlibrary{matrix,arrows}
\tikzstyle{every picture}+=[descr/.style={fill=white,inner sep=2.5pt},node distance=5.5em]
\newcommand{\obj}[3]{\node (#1) [#2] {$#3$};}
\newcommand{\arr}[3]{\path[->,font=\scriptsize](#2) edge node[auto] {$#1$} (#3);}

\newcommand{\tr}[1]{\stackrel{#1}{\longrightarrow}}

\newcommand{\Pfin}{\mathcal{P}_{\mathit{fin}}}
\renewcommand{\ker}{\mathit{ker}\,}
\newcommand{\quotient}[1]{_{/_{#1}}}

\newcommand{\name}[1]{\emph{(#1)}}

\newcommand{\innerslogan}[2]{\medskip\noindent\textbf{#1 #2}\medskip}
\newcommand{\slogan}[1]{\innerslogan{}{#1}}
\newcommand{\reminder}[1]{\innerslogan{Reminder: }{#1}}

\theoremstyle{plain}
\newtheorem{proposition}{Proposition}
\newtheorem{theorem}{Theorem}

\theoremstyle{definition}
\newtheorem{definition}{Definition}
\newtheorem{example}{Example}

\begin{document}

  \providecommand{\event}{ICE 2011}
  \providecommand{\titlerunning}{Interaction and observation, categorically}
  \providecommand{\authorrunning}{Vincenzo Ciancia}
  \title{\titlerunning}
  \author{\authorrunning\thanks{Research supported by the Netherlands Organization for
    Scientific Research VICI grant 639.073.501}
  \institute{Institute for Logic, Language and Computation \\ University of Amsterdam}}
  \maketitle

\begin{abstract}
This paper proposes to use \emph{dialgebras} to specify the semantics of interactive systems in a natural way. Dialgebras are a conservative extension of coalgebras. In this categorical model, from the point of view that we provide, the notions of observation and interaction are separate features. This is useful, for example, in the specification of process equivalences, which are obtained as kernels of the homomorphisms of dialgebras. As an example we present the asynchronous semantics of the CCS.
\end{abstract}

\section{Introduction}

The notions of \emph{interaction} and \emph{observation} play a key role in the semantics of concurrent and interactive systems. An \emph{interactive} system or process (imagine a web service, or an operating system) is typically not required to terminate, but it is not always equivalent to the deadlocked machine. This is because, along the execution of a system, the external environment is allowed to interact with the program and observe some side effects (typically, output from the system itself). 

However clear in principle, this intuition is lost whenever the semantics of an interactive system is modelled using \emph{labelled transition systems} (LTSs) or their categorical generalisation, the so-called \emph{coalgebras}. The reason is that every interaction that a system makes with the external world, be it originated from the environment, or from an internal action of the system itself, is described in the same way, as a transition from one state to the next.

In this work we turn our attention to a class of categorical models called \emph{dialgebras}. Dialgebras are a straightforward generalisation of both algebras and coalgebras. We interpret these models as a framework where one can describe separately the states of the system, the interactions that the environment and a process may have in each state, and the resulting observations. In our interpretation, dialgebras provide side-effecting operations, therefore providing both contexts and observations simultaneously.

The above is strongly reminiscent of the distinction between input and output in computer science. Thinking of interaction with the environment as an input to a process, and observation as its output, Mealy machines \cite{mea55} come to mind. These are functions $I \times X \to O \times X$, for $X$, $I$ and $O$ the set of states of the system, possible input values, and possible output values, respectively. It turns out that one of the simplest and more familiar examples of a dialgebra is a Mealy machine; in the same fashion, one of the simplest and more familiar examples of coalgebra is an LTS. This motivates the following slogan.

\slogan{Coalgebras generalise labelled transition systems; dialgebras generalise Mealy machines.}

As it happens with coalgebras w.r.t. LTSs, the merit of the generalisation is in the fact that, since dialgebras form a category, these generalised Mealy machines are now equipped with a standard notion of equivalence, which is given by the kernel of morphisms of the category. 

So, in our framework, the semantics of a programming language is given in terms of a dialgebra. The latter, as we will see, is a function $f$ from a set $\F X$ to a set $\B X$. $\F$ and $\B$ are parametrised in $X$, which is the set of states of a system. $\F$ describes a type of \emph{experiments} that an ideal observer can conduct. Then, results are observed, belonging to the set $\B X$ of possible \emph{observations}. The way to define the semantics is by choosing appropriate experiments and observations, and defining such a function $f$. From this information, using a small amount of category theory, a standard equivalence relation, called dialgebraic bisimilarity, is defined on $X$. Roughly speaking, two processes are dialgebraic bisimilar if they exhibit the same observations in the same experiments, and the states they reach after the experiments are bisimilar.

An example where it is useful to distinguish between interaction and observation is \emph{asynchronous} semantics. Asynchronous communication may be summarised by saying that ``the observer can not see the input actions of a process''. More precisely, the observer can not tell input actions from internal computations. In the dialgebraic perspective that we propose on asynchrony,  the observer can either sit and look at the system, seeing its output and internal computations, or try to send messages to it. However, a process can either read a message, or consume a message without actually reading it, and store it for later processing. The observer can not tell the two cases apart.

We provide a dialgebraic semantics of the asynchronous CCS, and prove that the obtained equivalence relation coincides with strong asynchronous bisimilarity. In this case, we make a distinction between an underlying \emph{operational semantics} which is expressed by the well-known LTS for the CCS, and the dialgebraic semantics, built on top of it, which specifies the semantic equivalence relation. Bisimilarity of the LTS of the operational semantics, which is also the \emph{synchronous} semantics, is not taken into account in the definition of the dialgebraic semantics.

Using a LTS is not necessary at all to specify a dialgebra. We do so mostly for the sake of simplicity: the asynchronous LTS semantics of process calculi is already well-understood. The operational semantics could in turn be defined as a dialgebra directly on the structure of processes (see \S \ref{sec:conclusions} for a brief discussion). On the other hand, the usage of a (however specified) operational semantics upon which a process equivalence is based can be considered at least a recurring pattern for the design of process equivalences. The definition of the semantic equivalence may be split in three steps, that we call \emph{execute}, \emph{interact}, \emph{observe}:
\begin{description}
 \item[execute:] the system is run by the means of its operational semantics, specifying some side effects of the process at each state of its execution;
 \item[interact:] the observer does experiments on the running system;
 \item[observe:] results are collected, allowing the observer to classify processes by how they react to experiments, giving rise to the behavioural equivalence of choice.
\end{description}

In coalgebras, these three steps are often tied to each other and not so easily separated. Dialgebras give us a different perspective on bisimilarity, where some actions are originated by a running process, and some others by the external environment.  The process and the environment may be very different, and the syntax of experiments is not (necessarily) the same as the syntax of processes. This is not so uncommon. Think e.g. of analysis or monitoring for security protocols. The entities (systems) that are being ``observed'' may be unknown machines or even human beings. The syntax of experiments conducted on such entities may have nothing in common with the entities themselves.
\medskip

\begin{example}
 For a classical example, think of an human (the observer) in front of a drink-vending machine. The observer can make experiments, such as pressing the buttons, inserting coins etc. A pre-condition for being able to tell something (and eventually get a drink) is that the machine is running. That is, a current state of the machine is defined, and the machine has an underlying \emph{operational semantics}, which is what the machine really does, independently from what the observer sees. While the machine is running, the observer performs its experiments, and observes some side-effects. The machines reaches a new state. This is an example where the ``syntax of experiments'' (e.g. inserting a coin, or pressing a button) is not the ``syntax of the vending machine'' which would be describing its internal mechanics.
\end{example}

\paragraph{Related work.} The study of dialgebras in computer science was initiated in \cite{Hag87} for the categorical specification of data types, and further investigated for the same purpose in \cite{PZ01}. So far, they have not been explored in detail. In this work we divert from the earlier research line: we find applications of dialgebras to programming language semantics, and look at the behavioural equivalences they induce on processes.  Moreover, even though we do not provide examples in the current paper, we do not restrict our attention just to the polynomial functors as the syntax of experiments (therefore, we use the equivalences from kernels of morphisms instead of the relational lifting used in \cite{PZ01}). This is since we expect that more complex functors may have useful applications (see \S \ref{sec:conclusions}). 

\paragraph{Map of the paper.} 

In \S \ref{sec:algebras-and-coalgebras} we give the definitions of algebras and coalgebras, for comparison with dialgebras. In \S \ref{sec:dialgebras} we give the definition of a dialgebra and explain their intended use. In \S \ref{sec:asynchronous-CCS} we present the asynchronous semantics of the CCS. In \S \ref{sec:dialgebras-ccs} we give a dialgebraic semantics to the CCS that coincides with the asynchronous one. In \S \ref{sec:further-examples} we informally discuss other examples of dialgebras. Finally in \S \ref{sec:conclusions} we sketch some possible future directions.

\section{Algebras and Coalgebras}\label{sec:algebras-and-coalgebras}

Algebras and coalgebras provide an established methodology for the specification of programming language syntax and semantics.  We give here a brief introduction to the definitions of algebra and coalgebra in a category, tailored to a comparison between these two constructions and that of a dialgebra. For more details and pointers to the rich existing literature on algebras and coalgebras, see \cite{Rut00}.

First we give the preliminary notion of a \emph{kernel}. For the category-theoretical concepts that we mention, we refer the reader to some basic category theory book (see e.g. \cite{Awo10}).

\begin{definition}
 The kernel of $f : X \to Y$ in a category $\C$ is the pullback (if it exists) of the diagram $f,f$.
\end{definition}

When $\C = \Set$, the kernel of $f$ (up-to isomorphism) is the set $\ker f = \{(x_1,x_2) \in X \times X \mid f(x_1) = f(x_2)\}$, equipped with the two obvious projections; this is an equivalence relation on $X$.

\begin{definition}\name{algebra}
 Given a endofunctor $\F$ in a category $\C$, an $\F$-algebra is a pair $(X,f : \F X \to X)$. An \emph{homomorphism} between two $\F$-algebras $(X,f)$ and $(Y,g)$ is an arrow $h : X \to Y$ such that $h \circ f = g \circ \F h$, that is, the following diagram commutes:
 \begin{center}
  \begin{tikzpicture}
     \obj {FX} {} {\F X}
     \obj {FY} {right of = FX} {\F Y}
     \obj {X} {below of = FX} {X}
     \obj {Y} {below of = FY} {Y}
     
     \arr {f} {FX} {X}
     \arr {g} {FY} {Y}
     \arr {\F h} {FX} {FY}
     \arr {h} {X} {Y}
  \end{tikzpicture}
 \end{center}
\end{definition}

When $\F$ is a \emph{polynomial} functor,  and $\C$ is $\Set$, then the notion of $\F$-algebra coincides with the classical notion of algebra for a signature (to recover the full power of equational specifications, one needs the stronger notion of algebra of a \emph{monad}, which is out of the scope of this discussion).

\reminder{algebras specify \emph{operations} on the elements of a set.}

For example, one can specify the signature (not the equations) of a monoid by providing a  set $X$ and the interpretation of composition and identity. In other words, a monoid can be regarded as an algebra for the functor $\F X = 1 + X\times X$, that is, a set $X$ and a function $f : 1 + (X \times X) \to X$. The function $f$ is the co-pairing of $f_1 : 1 \to X$, which is the interpretation of the identity of the monoid, and $f_\times : X \times X \to X$, which interprets composition.

Of particular relevance for programming language semantics is that algebras specify the \emph{abstract syntax} of programming languages, by providing operations on abstract syntax terms that can be applied to build larger terms. The functor $\F$ provides a syntax to describe operations on elements, and an algebra $(X,f)$ gives the semantics of such a syntax, by computing elements out of these operations.

\begin{definition}\name{coalgebra}
 Given an endofunctor $\B$ in a category $\C$, a $\B$-coalgebra is a pair $(X,f : X \to \B X)$. An \emph{homomorphism} between two $\B$-coalgebras $(X,f)$ and $(Y,g)$ is an arrow $h : X \to Y$ such that $\B h \circ f = g \circ h$, that is, the following diagram commutes:
 \begin{center}
  \begin{tikzpicture}
     \obj {X} {} {X}
     \obj {Y} {right of = X} {Y}
     \obj {BX} {below of = FX} {\B X}
     \obj {BY} {below of = FY} {\B Y}
     
     \arr {f} {X} {BX}
     \arr {g} {Y} {BY}
     \arr {h} {X} {Y}
     \arr {\B h} {BX} {BY}
  \end{tikzpicture}
 \end{center}
\end{definition}

A coalgebra in the category $\Set$ of sets and functions is a function $f : X \to \B X$ for some behavioural endofunctor $B : X \to X$. The action of $\B$ on objects yields a set $\B X$ for each $X$, which is intended to be the \emph{transition type} or \emph{observation type} of the system.

When $\B X = \Pfin(L \times X)$ and $\C$ is $\Set$, so that $X$ is a set, then a $\B$-coalgebra $f$ coincides with the classical notion of labelled transition system (LTS) with labels in $L$. Here, $X$ is the set of states of the system, $L$ is the set of labels, and for all $x \in X$, $f(x)$ is a set of \emph{labelled transitions}, that is, pairs $(\ell,x')$ consisting of a label and a destination state. 

\reminder{coalgebras specify \emph{observations} on the elements of a set.}

For example, one can specify an interactive system by providing a set $X$ of \emph{states}, and a transition function $f : X \to \Pfin (L \times X)$ describing the non-deterministic observations that we can make about the execution of a process, such as an input, an output, or an internal computation. It is useful to think of $L$, in this specific case, as the type of \emph{side effects} of the program execution.

The crucial fact about coalgebras is that they form a category, and the natural equivalence relation obtained by the kernel of homomorphisms generalises bisimilarity of LTSs.

By changing the transition type $\B$, one gains generality w.r.t. LTSs. For instance, one can use the probability distribution functor $\mathcal{D}$ in combination with other functors to express various degrees of probabilistic systems \cite{Sok05}.

\section{Dialgebras}\label{sec:dialgebras}

Behavioural equivalences, such as bisimilarity, are typically not based on the syntax of processes. Rather, an external \emph{observer} is assumed, that can see their behaviour. 
Processes are equivalent when the external observer can not tell them apart.

In this section we introduce dialgebras. We will see that the natural equivalence relation induced by morphisms is still based on behaviours. However, the external observer is now endowed with the power to interact with the system, by doing \emph{experiments} and \emph{observing} the results.

\begin{definition}\name{dialgebra} Given a category $\C$, and two endofunctors\footnote{In \cite{Hag87}, $\F$ and $\B$ just are required to have the same codomain, not to be endofunctors. The simplified definition we adopt is sufficient for this paper.} $\F, \B : \C \to \C$, a $(F,B)$-\emph{dialgebra} is a pair $(X,f)$ where $X$ is an object and $f : \F X \to \B X$ is an arrow of $\C$.
\end{definition}

We will just refer to such a structure as a \emph{dialgebra} when $\F$ and $\B$ are clear from the context. In the remainder of this section, let us fix two endofunctors $\F$ and $\B$.

We call $\F$ the \emph{interaction} functor, as it is intended to provide a syntax for constructing experiments. The functor $\B$ is the \emph{observation} functor, which is the type of the observed results.

\begin{definition}\name{dialgebra homomorphism}
 Given two dialgebras $(X,f)$ and $(Y,g)$, a \emph{dialgebra homomorphism} from $(X,f)$ to $(Y,g)$ is an arrow $h : X \to Y$ such that $g \circ \F h = \B h \circ f$, that is, the following diagram commutes
 \begin{center}
  \begin{tikzpicture}
     \obj {FX} {} {\F X}
     \obj {FY} {right of = FX} {\F Y}
     \obj {BX} {below of = FX} {\B X}
     \obj {BY} {below of = FY} {\B Y}
     
     \arr {f} {FX} {BX}
     \arr {g} {FY} {BY}
     \arr {\F h} {FX} {FY}
     \arr {\B h} {BX} {BY}
  \end{tikzpicture}
 \end{center}
\end{definition}

$(\F,\B)$-dialgebras and their homomorphisms form a category. Clearly, when $\B = \Id$ (the identity functor) one recovers the category of $\F$-algebras, and when $\F = \Id$ one recovers the category of $\B$-coalgebras. In this work, we only focus on dialgebras in the category $\Set$ of sets and functions.

\smallskip
\begin{example}
 Non-deterministic Mealy machines are dialgebras for the functors $\F X = I\times X$ and $\B X= \Pfin(O \times X)$, for $I$ the set of input values and $O$ the set of output values.
\end{example}
\smallskip

A dialgebra allows one to specify a set of experiments $\F X$ that, when executed trough $f$, give rise to observations in $\B X$. 
For a comparison, we mention \emph{bialgebras}. A bialgebra \cite{Tur97} is a pair $(f,g)$ of an algebra $f : \F X \to X$ and a coalgebra $g : X \to \B X$ having the same underlying set $X$. The algebra is used to construct elements, the coalgebra to observe them. Every bialgebra is also a dialgebra (the composite $g \circ f : \F X \to \B X$). Whereas a bialgebra specifies a set equipped with two separate, although possibly nicely interacting, coalgebraic and algebraic operations, a \emph{dialgebra} specifies a set equipped with operations that behave algebraically and  coalgebraically at the same time. The interpretation of the ``algebraic operations'' (the experiments) of a dialgebra does not yield a result, but rather an observation on it. When using dialgebras, just like in algebras, the observer can formally specify a structure (the experiment) that will be executed; just like in coalgebras, the observer interacts with the system in a step-wise fashion: at each state, an experiment can be conducted, yielding observations and possibly subsequent states, on which further experiments are possible.

\reminder{dialgebras specify \emph{operations} on the elements of a set, that yield \emph{observations} as a result.} 

We now define the underlying equivalence of a dialgebra.

\begin{definition}\label{def:dialgebraic-bisimilarity-kernel}\name{dialgebraic bisimilarity}
 Given a dialgebra $(X,f)$, dialgebraic bisimilarity is the relation $\approx \subseteq X\times X$ induced by the kernel of any homomorphisms $h : (X,f) \to (Y,g)$ on the underlying set $X$. That is, we say that $x \approx y \iff \exists (Y,g) . \exists h : (X,f) \to (Y,g) . h(x) = h(y)$.
\end{definition}

In the rest of the paper, we are going to see how to use dialgebras to model asynchrony. An example characterisation of the equivalence induced by morphisms as a back-and-forth condition, as typical in bisimilarity of LTSs, is given in Definition \ref{def:bisimilarity-of-dialgebras} and Theorem \ref{thm:kernel-equivalence}.

\section{The asynchronous CCS}\label{sec:asynchronous-CCS}

\subsection{Syntax and operational semantics}

The \emph{calculus of communicating systems} (CCS) \cite{Mil82} is a simple language for studying interactive systems, featuring interleaved parallel composition and synchronization over named channels. In this paper, we use the asynchronous semantics. The definitions we adopt come from the ones for the $\pi$-calculus in \cite{ACS98}; we refer the reader to that work for an in-depth study of asynchrony in process calculi.

Let $C$ denote a countable set of \emph{channels}. Define $L_i = C$, $L_{o} = \{ \bar c | c \in C\}$, $L_\tau = \{\tau\}$, $L=L_i \cup L_o \cup L_\tau$, the set of \emph{input labels}, \emph{output labels}, \emph{internal labels}, and \emph{labels}, respectively. These labels are observations on a system, representing sending ($\bar c$) or receiving ($c$) an input signal on a channel $c$, or doing an internal computation step $\tau$.

\begin{definition}\name{CCS syntax}
The syntax of the asynchronous CCS is defined by the following grammar, where $c$ ranges over a countable set $C$ of \emph{channel names}. 
 $$P ::= \emptyset \mid \tau.P \mid c . P \mid \bar c \mid P \parallel P  \mid P+Q$$
\end{definition}

We omit the replication and restriction constructs. This is done for ease of explanation as adding them does not affect our proofs. From now on, let $X$ denote the set of agents. In the syntax, $\emptyset$ represents the empty process, that does nothing; $\tau.P$ performs an internal computation step and then behaves as $P$; $c.P$ waits for an input signal on channel $c$, and then behaves as $P$; $\bar c$ sends an output signal on channel $c$; $P_1 \parallel P_2$ is the parallel composition of $P_1$ and $P_2$; $P+Q$ denotes non-deterministic choice. 

\begin{definition}\label{def:CCS-operational-semantics}\name{CCS operational semantics}
The operational semantics is given in the form of a LTS $t : X \to \Pfin(L \times X)$, defined by the following rules:
$$c . P \tr c P\,(in) \qquad \tau . P \tr \tau P \,(tau) \qquad \bar c \tr{\bar c} \emptyset \,(out)$$
$$\frac{P \tr \alpha P'}{P \parallel Q \tr \alpha P' \parallel Q} \,(par) \qquad  \frac{Q \tr \alpha Q'}{P \parallel Q \tr \alpha P \parallel Q'} \,(par') \qquad \frac{P\tr c P'\quad Q\tr {\bar c} Q'}{ P \parallel Q \tr \tau P' \parallel Q'} \,(syn)$$
$$\frac{P \tr \alpha P'}{P + Q \tr \alpha P'} \, (sum) \qquad \frac{Q \tr \alpha Q'}{P + Q \tr \alpha Q'} \, (sum')$$
\end{definition}

\medskip

\noindent Rules $(in)$, $(tau)$, and $(out)$ are straightforward. Rules $(par)$ and $(par')$ allow components to run in parallel in an interleaved fashion. Rule $(syn)$ allows a process that can do an input and a process that can do an output to synchronise. Rules $(sum)$ and $(sum')$ allow a non-deterministic choice to take place.

\subsection{Asynchronous bisimilarity}

We define asynchronous bisimulation and bisimilarity directly for CCS terms.

\medskip

\begin{definition}\label{def:asynchronous-bisimilarity}\name{CCS asynchronous bisimilarity}
 A relation $R \subseteq X \times X$ is an \emph{asynchronous simulation} if and only if, whenever $(x,y) \in R$, and $x \tr \alpha x'$, then there is $y'$ such that:
\begin{itemize}
 \item if $\alpha = \tau$ or $\alpha = \bar c$ for some $c$, then $y \tr \alpha y'$ and $(x',y') \in R$;

 \item if $\alpha = c$ for some $c$, then $\bar c \parallel y \tr \tau y'$ and $(x',y') \in R$ \\ or, equivalently\\ if $\alpha = c$ for some $c$, then $(x',y') \in R$ and either $y \tr c y'$ or $y \tr \tau y''$ with $y' = \bar c \parallel y''$.
\end{itemize}
An \emph{asynchronous bisimulation} is a simulation $R$ such that $R^{-1}$ is a simulation. \emph{Asynchronous bisimilarity} is the largest bisimulation.
\end{definition}

\medskip

We write $x \sim y$ whenever $x$ is asynchronous bisimilar to $y$, or equivalently there is some asynchronous bisimulation $R$ such that $(x,y) \in R$. In asynchronous bisimilarity, input labels can be matched ``loosely'' by a $\tau$ transition that stores an output process in parallel with the execution. We are going to see how to turn this definition into dialgebraic bisimilarity. Before that, we remark that synchronous bisimilarity (that would be obtained by employing strong bisimilarity on the LTS from Definition \ref{def:CCS-operational-semantics}) is included in the asynchronous one. The inclusion is strict. Two processes that are not synchronous bisimilar but are asynchronous bisimilar are $c . \bar c . \emptyset + \tau . \emptyset$ and $\tau . \emptyset$ (example adapted from \cite{ACS98}, where a thorough discussion can be found).

\section{Observing interactions}\label{sec:dialgebras-ccs}

Asynchronous bisimilarity does not coincide with the coalgebraic bisimilarity obtained from the transition system of Definition \ref{def:CCS-operational-semantics}. We define a dialgebra whose set of states is that of the CCS agents, and where dialgebraic bisimilarity is asynchronous bisimilarity. 

\subsection{Dialgebra for the asynchronous CCS}

First, we define, and fix hereafter, a specific pair of interaction and observation functors. 

\begin{definition}\name{CCS interaction and observation functors}\label{def:ccs-interaction-and-observation-functors}
 We let the interaction functor be $\F X = X + L_o \times X$, and the observation functor be $\B X = \Pfin((L_o\cup L_\tau) \times X)$.
\end{definition}

For any set $X$, an element $e$ of the disjoint union $\F X$ is either in the form $x$ or $(\bar c,x)$, for $c \in C$ and $x \in X$. Roughly, $e$ is the syntax of an experiment where we can either observe the execution of $x$, or send a signal to $x$ on channel $c$. An element $t$ of $\B X$ is a set of pairs $(\bar c,x')$ or $(\tau,x')$ for $c \in C$ and $x' \in X$. The element $t$ is a transition to $x'$ labelled with either the observation of an output signal on a certain channel, or of an internal computation step.  No input labels appear. Input is modelled as the argument of a function, instead of as a side-effect. This is in line with the idea that input is an action of the environment, not an action of the process.

We now define a $(\F,\B)$-dialgebra for the CCS. From now on, whenever $f$ is a dialgebra, we use the shorthand $e \tr \beta _f x'$ to denote that $(\beta,x') \in f(e)$, and omit $f$ when clear from the context.

\begin{definition}\label{def:CCS-dialgebraic-semantics}\name{CCS dialgebraic semantics}
 The $(\F,\B)$-dialgebra $f : \F X \to \B X$, where $X$ is the set of CCS processes equipped with the operational semantics of Definition \ref{def:CCS-operational-semantics}, is defined by the following rules:
 $$\frac{x \tr \alpha x'\quad \alpha = \tau \lor \alpha = \bar c}{x \tr \alpha _f  x'}\,(run) \qquad \frac{x \tr c x'}{(\bar c,x) \tr \tau _f x'}\,(in) \qquad \frac{x \tr \tau x'}{(\bar c,x)\tr \tau _f \bar c \parallel x'}\,(store)$$
\end{definition}

Premises of rules use the operational semantics of Definition \ref{def:CCS-operational-semantics}. Rule $(run)$ expresses the fact that we can observe the output and internal computation steps of a system. Rule $(in)$ states that whenever a process $x$ can do input, the experiment $(\bar c,x)$ yields the observation of an internal computation step. By Rule $(store)$, whenever a process can do an internal computation step, then it can also store an input signal from the environment for subsequent processing. The observations for the $(in)$ and $(store)$ rules are the same, therefore an observer can not distinguish the application of either one of the two rules. 

\subsection{Characterising dialgebraic bisimilarity}

A characterization of the equivalence induced by dialgebra homomorphisms for the functors $\F$ and $\B$ of Definition \ref{def:ccs-interaction-and-observation-functors} can be given as follows. 

\begin{definition}\label{def:bisimilarity-of-dialgebras}\name{Back-and-forth bisimilarity of dialgebras}
 Given a $(\F,\B)$-dialgebra $f : \F X \to \B X$, a relation $R \subseteq X \times X$ is a \emph{back-and-forth simulation} if and only if, for all $(x,y) \in R$ and $c \in C$:
  \begin{enumerate} 
    \item \label{def:bisimilarity-of-dialgebras:out-tau} whenever $x \tr \alpha _f x'$, there is $y'$ such that $y \tr \alpha _f y'$ and $(x',y') \in R$;
    \item \label{def:bisimilarity-of-dialgebras:in} whenever $(\bar c,x) \tr \tau _f x'$, there is $y'$ such that $(\bar c,y) \tr \tau _f y'$ and $(x',y') \in R$.
  \end{enumerate}
 A bisimulation is a simulation $R$ such that $R^{-1}$ is a simulation. Two elements of $X$ are said \emph{bisimilar} if and only if there is a bisimulation relating them. The corresponding relation is called bisimilarity.
\end{definition}

We write $x \simeq y$ to denote that $x$ is bisimilar to $y$.
 
\begin{proposition}\label{prop:bisimilarity-of-dialgebras-transitive}
 Back-and-forth bisimilarity is an equivalence relation.
\end{proposition}

\begin{theorem}\label{thm:kernel-equivalence}\name{back-and-forth vs. kernel}
When $\F$ and $\B$ are as in Definition \ref{def:ccs-interaction-and-observation-functors}, dialgebraic bisimilarity from Definition \ref{def:dialgebraic-bisimilarity-kernel} and back-and-forth bisimilarity from Definition \ref{def:bisimilarity-of-dialgebras} coincide.
\end{theorem}

\begin{proof}
 Fix a dialgebra $(X,f)$. First, consider a dialgebra $(Y,g)$ and $h : (X,f) \to (Y,g)$. We show that $\ker h$ is a back-and-forth bisimulation, therefore it is included in $\simeq$. Assume $h x = h y$ for some $x,y \in X$. For all $\alpha \in L$, by definition of homomorphism, we have 
 $g(\F h(\alpha,x)) = \B h (f (\alpha,x))$. Therefore $g(\alpha, h y) = \B h (f (\alpha,x))$. Let $(\beta,x') \in f(\alpha,x)$. Then $(\beta,h x') \in \B h (f(\alpha,x))$, therefore $(\beta,h x') \in g(\alpha, h y) = g(\F h (\alpha,y))$, thus by commutativity $(\beta,h x') \in \B h (f(\alpha,y))$. Then there is some $y'$ such that $(\beta,y') \in f(\alpha,y)$ and $h x' = h y'$. This proves that $\ker h$ is a simulation. Notice that the kernel of a function is an equivalence relation, therefore $(\ker h)^{-1}=\ker h)$, thus proving that $\ker h$ is a bisimulation.
 For the other direction of the proof, let $[x]$ denote the equivalence class of $x$ in $X\quotient{\simeq}$. Consider the quotient dialgebra $(X\quotient{\simeq},f\quotient{\simeq})$, with $f\quotient{\simeq}(\alpha,[x]) = \{(\beta,[x']) | (\beta,x') \in f(x) \}$. Notice that $f\quotient\simeq$ is well defined by definition of $\simeq$. The quotient function $h x=[x]$ is obviously a homomorphism of dialgebras, and it is the case that whenever $x \simeq y$ then $h(x) = h(y)$.
\end{proof}

Finally, we prove that asynchronous and back-and-forth bisimilarity coincide.

\begin{theorem}\name{asynchronous vs. back-and-forth}
 Asynchronous bisimilarity from Definition \ref{def:asynchronous-bisimilarity} and back-and-forth bisimilarity coincide for the set $X$ of CCS agents, that is: for all $x,y \in X$, we have $x \sim y$ if and only if $x \simeq y$. Therefore, by Theorem \ref{thm:kernel-equivalence}, asynchronous bisimilarity and dialgebraic bisimilarity coincide.
\end{theorem}

\begin{proof}
 We provide the proof just for completeness, as it is immediate from the characterisation of asynchronous bisimilarity as a $1$-bisimilarity in \cite{ACS98}. We prove that $\sim$ is a back-and-forth bisimulation. Symmetry, and Case \ref{def:bisimilarity-of-dialgebras:out-tau} from Definition \ref{def:bisimilarity-of-dialgebras} are obvious. For Case \ref{def:bisimilarity-of-dialgebras:in}, suppose $(\bar c,x) \tr{\tau} x'$. Then we distinguish two cases.
 \begin{itemize} 
  \item if Rule $(in)$ is applied to $(\bar c,x)$, we have $x \tr c x'$. We now look at Definition \ref{def:asynchronous-bisimilarity}. Since $x\sim y$, we have $\bar c \parallel y \tr{\tau} y'$ with $x' \sim y'$. We inspect the rules in Definition \ref{def:CCS-operational-semantics}. The rules that can be applied to $\bar c \parallel y$ are $(par)$ and $(syn)$ (and $(par')$ which is treated in the same way as $(par)$). 
  Therefore we have either $y \tr \tau y''$ with $y' = \bar c \parallel y''$, or $y \tr c y'$. By applying either Rule $(in)$ or $(store)$ from Definition \ref{def:CCS-dialgebraic-semantics}, we obtain $(\bar c, y) \tr{\tau} _f y'$ and since $x' \sim y'$ we get the thesis.

  \item if Rule $(store)$ is applied to $(\bar c,x)$, then $x \tr \tau x''$ with $x' = \bar c \parallel x''$. Therefore, $y \tr \tau y''$ and $x'' \sim y''$. It is well known and easy to prove that $x'' \sim y'' \implies \bar c \parallel x'' \sim \bar c \parallel y''$. Therefore by applying Rule $(store)$ we get $(\bar c,y) \tr \tau y'$ and $x'\sim y'$, q.e.d.
 \end{itemize}
 Next, we prove that $\simeq$ is an asynchronous bisimulation. Suppose $x \simeq y$ and $x \tr \alpha x'$. We look at Definition \ref{def:asynchronous-bisimilarity}. The cases for $\alpha = \tau$ or $\alpha = \bar c$ are obvious. Suppose $\alpha = c$ for some $c$. By Rule $(in)$ in Definition \ref{def:CCS-dialgebraic-semantics} we have $(\bar c,x) \tr \tau _f x'$ and by $x \simeq y$ we get $(\bar c,y) \tr \tau _f y'$ with $x' \simeq y'$. Either Rule $(in)$ or $(store)$ from Definition \ref{def:CCS-dialgebraic-semantics} can be applied to $(\bar c,y)$. Therefore either $y \tr c y'$, or $y \tr \tau y''$ with $y' = \bar c \parallel y''$. In both cases, we have $\bar c \parallel y \tr \tau y'$ and $x' \simeq y'$, from which the thesis.
\end{proof}

\section{Discussion on further examples}\label{sec:further-examples}

The example that we present is very simple, and purposed to illustrate just the idea of an observer that can interact with the examined system. More interesting dialgebras can be described by either moving to a richer category than $\Set$, or by changing the interaction and observation functor. We briefly describe some possible constructions, whose detailed study is left for future work.

\paragraph{Complex systems} Consider dialgebras of the form $f : \Pfin(X) \to L \times \Pfin(X)$. At each step in time, from a set $p \in \Pfin(X)$, a side effect in $L$ is observed, and a new set of elements $p'$ is obtained. Such a function may be used to represent systems where the semantics depends on a number of entities that collaborate. At each step in time, the system evolves, some old elements may be ``destroyed'' and new elements can be created, while some side effect in $L$ takes place. The behaviour of the system is \emph{more than the sum of its parts}, in the sense that it is not determined by the behaviour of singletons. The semantics of $\{x\}$, that is, $x$ in isolation, may be totally unrelated to the semantics of, say, the set $\{x,y\}$. Notice that $f : \Pfin(X) \to L \times \Pfin(X)$ is also a coalgebra in $\Set$ for the functor $\mF{T}(X) = L \times X$, having $\Pfin(X)$ as underlying set. However, it's obvious that the obtained notion of bisimulation is not the same, even by just looking at types. Seeing $f$ as a coalgebra, one gets a relation on $\Pfin(X)$; seeing it as a dialgebra, one gets a relation on $X$, that takes into account how elements behave when joined to the same sets of other elements.

\paragraph{Chemical reactions} In many cases programming language semantics has been inspired by chemical and biological processes. Consider the finite multi-set functor 
$\mathcal{M}(X)=\{m : X \to \mathbb{N} \mid \{x \mid m(x) \neq 0\} \text{ is finite}\}$. Think of $X$ as a set of elements that take part in \emph{reactions} in variable quantities. A dialgebra $ f : \mathcal{M}(X) \to \mathcal{M}(X)$ specifies how a given reaction evolves by creating a multi-set of products from a multi-set of reagents. The obtained notion of bisimilarity makes reagents equivalent when substituting one with the other in any reaction yields equivalent products, in the same quantities.

\paragraph{The $\pi$-calculus} A very similar development to the one presented here, exemplifying the use of a different base category, is the semantics of the asynchronous $\pi$-calculus. 
Similarly to what happens for the synchronous pi-calculus and coalgebras \cite{ft99}, one would use the functor category $\Set^\I$, where $\I$ is the category of  of finite sets and injections. The semantics would involve the endofunctor for fresh name allocation $\delta$ which is typical of functor categories, which is needed to properly model bound output. Dialgebras using $\delta$ correspond to Mealy machines with name allocation along output, whose study is possibly of interest independently from the specific application of the $\pi$-calculus.

\paragraph{Testing semantics}  Even though we spoke of interaction and observation, we did not mention so far the family of  \emph{testing equivalences} (see \cite{DH84}), where interaction and observation play a key role. Testing equivalences are defined as those obtained by putting a process in parallel with an arbitrary other process making use of a distinguished channel. Output on such channel signals that a test has been successful. Binary dialgebras come to mind as an effective way to represent such kinds of equivalence relations. However, in testing  equivalences, one is not able to observe how many synchronisation steps between processes are needed before the success signal is sent. Such a semantics could be defined by observing the behaviour of a process as a single ``big step''; however, this would defeat the implicit coinductive properties of dialgebras. A common feature of dialgebras and coalgebras is that observations lead to successor states, and then in a coinductive fashion further experiments/observations can be done on these successor states. However, in the case of testing equivalences, there is no successor state: once success is signalled, the experiment is concluded. Further investigation may yield non-obvious coinductive ways to represent these kind of relations on processes.

\section{Conclusions and future work}\label{sec:conclusions}

The construction we have seen in \S \ref{sec:dialgebras-ccs} has obvious similarities with barbed equivalence and with the asynchronous semantics of the $\pi$-calculus by Honda and Tokoro (both described in \cite{ACS98}). That's expectable, since in the end we are trying to describe the same equivalence relation. 

In the case of the asynchronous CCS, it is not difficult to recover a coalgebraic semantics. This is done by translating the dialgebraic semantics along the isomorphisms $X + L_o \times X \to \Pfin((L_o + L_\tau) \times X) \cong (L_i + 1) \times X \to \Pfin((L_o + L_\tau) \times X) \cong X \to (\Pfin((L_o + L_\tau) \times X))^{L_i + 1}$ (indeed, after noting that $L_i \cong L_o$). Notice that the latter is genuinely a coalgebra for the functor $(\Pfin(L_o + L_\tau) \times -)^{L_i + 1}$. It is not difficult to see that such a translation preserves and reflects the equivalence induced by kernels of homomorphisms (of dialgebras in one case, of coalgebras in the other). 

Even though it might be interesting to derive a coalgebraic semantics for the asynchronous CCS, we do not discuss the details of such a construction: the purpose of using this language as an example is not to provide a new semantics for asynchronous process calculi. Rather, the asynchronous CCS is possibly the simplest language where it makes sense to distinguish between moves of the environment and moves of the system being examined in order to define the semantics. Our aim is to show how such a distinction is naturally encoded using dialgebras, and their built-in definition of behavioural equivalence makes them appealing as an alternative to coalgebras in the specification of interactive systems.

We summarise below some possible future directions and open questions.

\paragraph{Inductively defined dialgebras.}
We defined a dialgebra for the asynchronous CCS by assuming an existing operational semantics. It is indeed possible to specify such a semantics using dialgebras. First, because coalgebras actually \emph{are} dialgebras with $\F = \Id$. Moreover, one could easily define an $(\F,\B)$-dialgebra, for $\F$ and $\B$ as in \S \ref{sec:dialgebras-ccs}, directly by induction on terms forming the set of agents $X$, in the same fashion of bialgebras and distributive laws. It would be relevant to study distributive laws and specification languages for inductively defined dialgebras, following the same route of bialgebras. Doing so, it would be possible to guarantee that a given dialgebraic semantics of a calculus is also a congruence with respect to the operators of the algebra describing its syntax.

\paragraph{Logics}
Dialgebras are equipped in \cite{PZ01} with \emph{dialgebraic specifications}, even though neither a full adequacy result relating logical equivalence and bisimilarity, nor Birkoff-style theorems are established. It ought to be clarified what is a logical formalism that adequately specifies dialgebras. Such a logic would be an intermediate language between modal and equational logic. The work \cite{Pal2002}, relating dialgebras to the so-called \emph{abstract logics} is possibly relevant. This research line should take advantage of, and extend, the many existing studies in the field of coalgebraic modal logic. 

\paragraph{Non-polynomial interaction functors}
Dialgebras are parametrised in the interaction and observation functors. Non-polynomial interaction functors, such as e.g. a probability distribution over the input values, could provide valuable case studies. Modulo the observation functor being ``probabilised'', too,  such dialgebras may be used to represent a kind of probabilistic Mealy machines, where the probability distribution of the input determines that of the output. It should be understood whether in the case of non-polynomial interaction functors there is some gain in expressive power w.r.t. coalgebras.

\paragraph{Minimisation}
Coalgebras have an elegant and simple minimisation procedure, based on \emph{iteration along the terminal sequence} and generalising partition refinement for automata. Are there canonical models in dialgebras? The results in \cite{PZ01} seem to point out that such a theory would be very difficult in the presence of so-called \emph{binary methods}, due to non-closure of bisimulations under union, and the lack of a final dialgebra. However, the (dialgebraic) bisimilarity quotient may still exist in interesting cases. More work is required on this side. The precise conditions when final dialgebras and bisimilarity quotients exist should be clarified. Also notice that in \cite{PZ01} $\F$ is assumed to be polynomial. Since we seek for non-polynomial interaction functors too, we expect that some work on the side of canonical models will be needed in order to understand how bisimilarity of dialgebra can be decided, possibly by finite representations derived from the definitions of $\F$ and $\B$.


\end{document}